%% file: main.tex
\documentclass[journal]{IEEEtran}

\usepackage{amsmath,amssymb, graphicx}
\usepackage[english]{babel}
\usepackage{algorithm}
\usepackage{paralist} 
\usepackage{color}
\usepackage{multirow}
\usepackage{xspace}
\usepackage{tabularx}
\usepackage{rotating}
\usepackage{silence}
\usepackage{packages/algpseudocodex}
\usepackage{hyperref}
\usepackage[mathscr]{eucal} 
\usepackage{amsbsy} 
\usepackage{subcaption}
\usepackage{booktabs}
\usepackage{xcolor}
\usepackage{tikz}
\usepackage{multirow}
\usepackage{epstopdf}
\usepackage{balance}
\usepackage{tablefootnote}
\usepackage{empheq} 
\usepackage{moresize}
\usepackage{enumitem}
\usepackage{graphicx}
\usepackage[utf8]{inputenc}
\setlist{nolistsep}
\usepackage{layouts}
\usepackage{float}
\DeclareCaptionType{equ}[][]

\usepackage[font=small,labelfont=bf,skip=0pt]{caption}

\usepackage{amsthm}
\usepackage{braket}
\usepackage{hyperref}

\newtheorem{lemma}{Lemma}
\newtheorem{theorem}{Theorem}
\newtheorem{corollary}{Corollary}
\newtheorem{definition}{Definition}

\DeclareCaptionType{copyrightbox}
\usepackage{url}

\renewcommand{\algorithmicrequire}{\textbf{Input:}}
\renewcommand{\algorithmicensure}{\textbf{Output:}}
\newcounter{ALC@tempcntr}
\input{dfn}

\begin{document}



\author{
  Michelle L. Wu\textsuperscript{1,2}\thanks{
    ORCID: \href{https://orcid.org/0000-0002-1881-1368}{0000-0002-1881-1368} \\
    \textsuperscript{1}NEXQ Inc., \texttt{hq@nexq.us} \\
    \textsuperscript{2}Carnegie Mellon University, Tepper School of Business, \texttt{mlwu@andrew.cmu.edu} \\
  }
}


\title{A Grover-Based Quantum Algorithm for Solving Perfect Mazes via Fitness-Guided Search}
\maketitle
%
\input{proposal}



\bibliographystyle{IEEEtran}
\bibliography{bib/proposal}

\end{document}

%% file: dfn.tex
\newcommand{\hide}[1]{}

\newcommand{\ben}{\begin{enumerate*}}
\newcommand{\een}{\end{enumerate*}}
\newcommand{\bit}{\begin{itemize*}}
\newcommand{\eit}{\end{itemize*}}



%% file: proposal.tex
\begin{abstract}
We present a quantum algorithm for solving perfect mazes by casting the pathfinding task as a structured search problem. Building on Grover's amplitude amplification, the algorithm encodes all candidate paths in superposition and evaluates their proximity to the goal using a reversible fitness operator based on quantum arithmetic. A Grover-compatible oracle marks high-fitness states, and an adaptive cutoff strategy refines the search iteratively. We provide formal definitions, unitary constructions, and convergence guarantees, along with a resource analysis showing efficient scaling with maze size and path length. The framework serves as a foundation for quantum-hybrid pathfinding and planning. The full algorithmic pipeline is specified from encoding to amplification, including oracle design and fitness evaluation. The approach is readily extensible to other search domains, including navigation over tree-like or acyclic graphs.
\end{abstract}

\section{Introduction} \label{sec:intro}

Maze solving represents a fundamental computational problem with applications spanning robotics, game theory, and artificial intelligence. In its general form, the task involves identifying a valid path from a designated start cell to a goal cell within a discretized spatial environment. When constrained to perfect two-dimensional mazes, defined as acyclic, fully connected grids, the problem remains computationally intractable and belongs to the class of NP-complete problems. Classical approaches such as backtracking and breadth-first search exhibit exponential scaling behavior with respect to maze size and branching factor, limiting their practical applicability in large instances.

Quantum computing offers an alternative paradigm by leveraging principles of superposition and interference to achieve algorithmic speedups. In particular, Grover's search algorithm provides a quadratic improvement for unstructured search problems and serves as a foundation for formulating structured search tasks in the quantum domain.

This work presents a complete quantum algorithm for solving perfect mazes by framing the path-finding problem as a structured quantum search. The proposed framework incorporates the following key components:

\begin{itemize}
    \item Quantum encoding of all candidate paths in a uniform superposition;
    \item A reversible fitness evaluation operator that scores paths based on proximity to the goal state;
    \item A Grover-compatible oracle that selectively amplifies high-fitness paths;
    \item An iterative amplitude amplification scheme that progressively refines the solution space.
\end{itemize}

While previous studies have explored quantum approaches to maze solving at a conceptual level, such as explicit circuit constructions for fitness evaluation and oracle design, these have not been rigorously defined. This work addresses these gaps by providing formal mathematical definitions, circuit-level implementations, and a complete convergence analysis of the quantum search process. Additionally, we evaluate the time complexity and resource requirements of the algorithm, highlighting its feasibility within the context of near-term quantum architectures.

\section{Problem Formalization} \label{sec:problem}

Let the maze be modeled as a perfect $m \times m$ grid, i.e., a two-dimensional discrete space:
\begin{equation}
    \mathcal{M} := \{0, 1, \dots, m-1\} \times \{0, 1, \dots, m-1\}
\end{equation}
Each cell $(i, j) \in \mathcal{M}$ may have open walls in one or more of the four cardinal directions: $\{N, E, S, W\}$. The maze is:
\begin{itemize}
    \item \textit{Solvable:} there exists at least one path from a unique start cell $(i_s, j_s)$ to the goal cell $(i_f, j_f)$,
    \item \textit{Loop-free:} the maze is a tree over $\mathcal{M}$, i.e., there is exactly one simple path between any two reachable cells.
\end{itemize}

\subsection*{Path Representation}

A \textit{path} of length $n$ is a sequence of $n$ directional moves:
\[
P = (d_1, d_2, \dots, d_n), \quad \text{where } \quad d_k \in \{N, E, S, W\}
\]
Each direction can be encoded as a 2-bit binary:
\[
N = 00,\quad E = 01,\quad S = 10,\quad W = 11
\]
Thus, a path of length $n$ corresponds to a $2n$-bit string, or equivalently, a quantum state on $2n$ qubits:
\begin{equation}
    \ket{P} = \ket{d_1} \otimes \ket{d_2} \otimes \cdots \otimes \ket{d_n} \in \left( \mathbb{C}^2 \right)^{\otimes 2n}
\end{equation}

\subsection*{Search Space}

The full set of candidate paths of length $n$ is:
\[
\mathcal{P}_n := \{N, E, S, W\}^n, \quad \text{so that } \quad |\mathcal{P}_n| = 4^n
\]
Each $P \in \mathcal{P}_n$ may or may not result in a valid traversal through the maze from the start cell. We define a transition function that simulates each path.

\subsection*{Maze Transition Function}

Let $\delta : \mathcal{M} \times \{N,E,S,W\} \to \mathcal{M} \cup \{\bot\}$ be the partial transition function:
\begin{equation}
    \delta\big( (i,j), d \big) =
\begin{cases}
(i', j') & \text{if move in $d$ from $(i,j)$ is legal} \\
\bot & \text{if move hits a wall or exits the grid}
\end{cases}
\end{equation}
Then, a path $P = (d_1, \dots, d_n)$ induces a sequence of positions:
\[
(i_0, j_0) := (i_s, j_s), \quad (i_k, j_k) := \delta\big( (i_{k-1}, j_{k-1}), d_k \big)
\]
If at any step $k$, we have $(i_k, j_k) = \bot$, the path is considered invalid past that point.

\subsection*{Goal Definition}

Let the final position of path $P$ be:
\begin{equation}
    \text{end}(P) := (i_n, j_n)
\end{equation}
where $(i_n, j_n)$ is computed by sequential application of $\delta$. If $\delta$ fails at step $k < n$, then $\text{end}(P)$ is the last valid cell prior to $\bot$.

The optimization objective is to find a path $P^\star \in \mathcal{P}_n$ that:
\begin{enumerate}
    \item Exactly reaches the goal: $\text{end}(P^\star) = (i_f, j_f)$, or
    \item Minimizes a known metric (e.g., Manhattan distance) to the goal:
    \begin{equation}
        P^\star := \arg\min_{P \in \mathcal{P}_n} \text{dist}(\text{end}(P), (i_f, j_f))
    \end{equation}
\end{enumerate}
This is formalized in the fitness function described in Section~\ref{sec:fitness}.

\subsection*{Path Length Parameter}

The choice of $n$ is crucial:
\begin{itemize}
    \item If $n < \text{length}(P^*)$, the solution space excludes the true path.
    \item If $n > \text{length}(P^*)$, the space includes redundant or invalid paths.
\end{itemize}
We assume $n \geq \ell_{\min}$, where $\ell_{\min}$ is the unique path length from $(i_s, j_s)$ to $(i_f, j_f)$. This value may be set heuristically or adapted in an iterative process (see Section~\ref{sec: iterative}). We note that because the maze is loop-free, it induces a spanning tree over $\mathcal{M}$ for the tree-based maze representation. Thus, the set of all valid paths corresponds to root-to-leaf paths in this tree. This structure is beneficial for pruning during classical or quantum search.

\section{Quantum Encoding and Superposition} \label{sec:encoding}

Let a path $P = (d_1, d_2, \dots, d_n) \in \mathcal{P}_n$ be a sequence of $n$ moves, where each $d_k \in \{N, E, S, W\}$. We define a binary encoding map:
\begin{equation}
    \chi: \{N, E, S, W\} \to \{0,1\}^2
\quad \text{such that} \quad
\begin{aligned}
\chi(N) &= 00 \\
\chi(E) &= 01 \\
\chi(S) &= 10 \\
\chi(W) &= 11
\end{aligned}
\end{equation}
Then, a path $P \in \mathcal{P}_n$ is mapped to a $2n$-bit binary string:
\begin{equation}
    \chi(P) = \chi(d_1) \, \| \, \chi(d_2) \, \| \, \cdots \, \| \, \chi(d_n) \in \{0,1\}^{2n}
\end{equation}

\subsection*{Quantum Register Representation}

Each bit is associated with a qubit, and the full quantum register lives in the Hilbert space:
\begin{equation}
    \mathcal{H} := \left( \mathbb{C}^2 \right)^{\otimes 2n}, \quad \dim \mathcal{H} = 2^{2n}
\end{equation}
For each path $P \in \mathcal{P}_n$, define:
\begin{equation}
    \ket{P} := \ket{\chi(P)} \in \mathcal{H}
\end{equation}
For example, if $P = (N, E, S)$, then:
\[
\ket{P} = \ket{00\,01\,10} = \ket{0} \otimes \ket{0} \otimes \ket{0} \otimes \ket{1} \otimes \ket{1} \otimes \ket{0}
\]

\subsection*{Uniform Superposition Over Paths}

We aim to construct the state:
\begin{equation}
    \ket{\Omega} := \frac{1}{\sqrt{|\mathcal{P}_n|}} \sum_{P \in \mathcal{P}_n} \ket{P}
= \frac{1}{2^n} \sum_{P \in \mathcal{P}_n} \ket{\chi(P)}
\end{equation}
This is a uniform superposition over all encoded direction sequences of length $n$.

\subsection*{Hadamard Preparation and Encoding Validity}

Let the initial state be:
\begin{equation}
    \ket{0}^{\otimes 2n} \in \mathcal{H}
\end{equation}
Apply the Hadamard gate to each qubit:
\begin{equation}
    H^{\otimes 2n} \ket{0}^{\otimes 2n} = \frac{1}{\sqrt{2^{2n}}} \sum_{z \in \{0,1\}^{2n}} \ket{z}
\end{equation}
This yields a uniform superposition over all binary strings of length $2n$. However, not all such strings correspond to valid direction sequences. In fact, the subset:
\begin{equation}
    S := \chi(\mathcal{P}_n) \subset \{0,1\}^{2n}
\end{equation}
has size exactly $|S| = 4^n$ — corresponding to strings where every adjacent pair of bits encodes a direction.

To restrict the uniform superposition to valid encodings, we group qubits into $n$ pairs:
\begin{equation}
    \mathcal{H} = \bigotimes_{k=1}^n \mathcal{H}_k, \quad \text{where } \mathcal{H}_k := \mathbb{C}^2 \otimes \mathbb{C}^2
\end{equation}
On each 2-qubit block, apply:
\begin{equation}
    H^{\otimes 2} \ket{00} = \frac{1}{2} \sum_{b \in \{0,1\}^2} \ket{b}
\end{equation}
Then the full state becomes:
\begin{equation}
    \ket{\Omega} = \bigotimes_{k=1}^n \left( \frac{1}{2} \sum_{b_k \in \{0,1\}^2} \ket{b_k} \right)
= \frac{1}{2^n} \sum_{P \in \mathcal{P}_n} \ket{\chi(P)}
\end{equation}

\begin{theorem}[Superposition Correctness]
Let $n \in \mathbb{N}$. Let $\mathcal{P}_n = \{N, E, S, W\}^n$ be the set of all direction sequences of length $n$, and define the encoding map $\chi: \mathcal{P}_n \to \{0,1\}^{2n}$ as above.

Let $\ket{\Omega} \in \left( \mathbb{C}^2 \right)^{\otimes 2n}$ be the quantum state defined by:
\begin{equation}
    \ket{\Omega} := \bigotimes_{k=1}^n \left( H^{\otimes 2} \ket{00} \right)
\end{equation}
Then:
\begin{equation}
    \ket{\Omega} = \frac{1}{2^n} \sum_{P \in \mathcal{P}_n} \ket{\chi(P)}
\end{equation}
That is, $\ket{\Omega}$ is a uniform superposition over all valid direction sequences of length $n$.
\end{theorem}

\begin{proof}
Let $\mathcal{H} := \left( \mathbb{C}^2 \right)^{\otimes 2n}$, and consider its decomposition as:
\begin{equation}
    \mathcal{H} = \bigotimes_{k=1}^n \mathcal{H}_k, \quad \text{where } \mathcal{H}_k = \mathbb{C}^2 \otimes \mathbb{C}^2
\end{equation}
Each 2-qubit subsystem $\mathcal{H}_k$ will encode one direction $d_k \in \{N,E,S,W\}$ via the injective map $\chi$.

Now apply $H^{\otimes 2}$ to each pair of qubits initialized in $\ket{00}$:
\begin{equation}
    H^{\otimes 2} \ket{00} = \frac{1}{2} \sum_{b \in \{0,1\}^2} \ket{b} = \frac{1}{2} (\ket{00} + \ket{01} + \ket{10} + \ket{11})
\end{equation}
This forms a uniform superposition over the four valid direction encodings. Since all $n$ direction encodings are independent, we take the tensor product:
\begin{align}
    \ket{\Omega} &= \bigotimes_{k=1}^n \left( \frac{1}{2} \sum_{b_k \in \{0,1\}^2} \ket{b_k} \right)\\&= \frac{1}{2^n} \sum_{(b_1, \dots, b_n) \in (\{0,1\}^2)^n} \ket{b_1 \| b_2 \| \cdots \| b_n}
\end{align}

Each concatenated string $b_1 \| \cdots \| b_n \in \{0,1\}^{2n}$ is equal to $\chi(P)$ for some $P \in \mathcal{P}_n$ by definition of $\chi$. The sum therefore runs over all such encodings:
\begin{equation}
    \ket{\Omega} = \frac{1}{2^n} \sum_{P \in \mathcal{P}_n} \ket{\chi(P)}
\end{equation}

Thus, the proof that $\ket{\Omega}$ is a uniform superposition over all direction sequences of length $n$.
\end{proof}

\section{Fitness Operator Construction} \label{sec:fitness}

The fitness function evaluates how close a given path $P \in \mathcal{P}_n$ ends to the maze goal cell $(i_f, j_f)$ starting from $(i_s, j_s)$. Each direction $d_k \in \mathcal{D} = \{N, E, S, W\}$ is encoded using 2 qubits:
\[
N = 00, \quad E = 01, \quad S = 10, \quad W = 11.
\]
A path $P$ of length $n$ corresponds to a $2n$-qubit register $\ket{x} = \ket{\chi(P)}$.

We define the fitness of $P$ using squared Euclidean distance:
\begin{equation}
    \text{fitness}(P) = C - \left[(i_{\text{end}} - i_f)^2 + (j_{\text{end}} - j_f)^2\right]
\end{equation}
Here, $C = 2^{r}$ is a power-of-two constant satisfying $C \geq 2(m - 1)^2$, ensuring that fitness is always non-negative and maximized when the path reaches the goal. The smallest such $r \in \mathbb{N}$ determines the number of qubits required to store the fitness value.

\subsection*{Definition: Fitness Operator}

\begin{definition}[Fitness Operator]
Let $\ket{x}$ be the quantum register encoding path $P \in \mathcal{P}_n$, and let $\ket{0}_f$ be an ancillary fitness register of $r$ qubits. The fitness operator is the unitary transformation:
\begin{equation}
    F: \ket{x} \ket{0}_f \mapsto \ket{x} \ket{\text{fitness}(x)}
\end{equation}
where the fitness is computed via a reversible simulation of the path defined by $\ket{x}$ and stored in $\ket{f}$.
\end{definition}

\subsection*{Quantum Circuit Implementation}

Let us denote the full register configuration:

\begin{itemize}
    \item $\ket{x}$: $2n$-qubit path register
    \item $\ket{i}, \ket{j}$: Position registers, each of $\lceil \log_2 m \rceil$ qubits
    \item $\ket{d}$: Distance register (stores squared Euclidean distance)
    \item $\ket{f}$: Fitness register of $r = \lceil \log_2 C \rceil$ qubits, where $C \geq 2(m - 1)^2$
    \item Ancillae: Temporary registers for subtraction, squaring, and cleanup
\end{itemize}

The circuit for $F$ consists of the following stages.

\subsubsection*{Path Simulation}

For each $k = 1$ to $n$, extract the 2-qubit slice $\ket{d_k} \subset \ket{x}$ and apply a direction-controlled update:
\begin{equation}
    \delta(i,j; d_k) :=
\begin{cases}
i \mapsto i - 1 & \text{if } d_k = 00 \\
j \mapsto j + 1 & \text{if } d_k = 01 \\
i \mapsto i + 1 & \text{if } d_k = 10 \\
j \mapsto j - 1 & \text{if } d_k = 11
\end{cases}
\end{equation}
Each update is implemented via multi-controlled add/subtract circuits (see \cite{cuccaro2004adder}) using Toffoli and CNOT gates. These operations are fully reversible.

\subsubsection*{Distance Computation}

Let the final position be $(i,j)$. The squared Euclidean distance is:
\[
d = (i - i_f)^2 + (j - j_f)^2
\]
Use quantum subtractors to compute:
\begin{equation}
    \ket{d_i} := \ket{i - i_f}, \quad \ket{d_j} := \ket{j - j_f}
\end{equation}
Then apply reversible squaring circuits (e.g., Toffoli-based multipliers as in \cite{draper2000addition}):
\begin{equation}
    \ket{d_i^2}, \ket{d_j^2} \Rightarrow \ket{d} = \ket{d_i^2 + d_j^2}
\end{equation}

\subsubsection*{Fitness Calculation}

Compute:
\begin{equation}
    \ket{f} = \ket{C - d}
\end{equation}
using a reversible subtractor. The constant $C = 2^{r}$ is embedded as a fixed binary input in the arithmetic circuit. The result is written into $\ket{f}$.

\subsubsection*{Uncomputation}

All ancilla registers used for intermediate subtraction and squaring are uncomputed by reversing the prior operations. This yields:
\begin{equation}
    \ket{x} \ket{f}
\end{equation}
ensuring the operator is reversible and ready for integration into the larger Grover-style quantum circuit.

\subsection*{Theorem: Fitness Operator is Unitary}

\begin{theorem}
The fitness operator $F$ defined above is unitary on the joint space of the path and fitness registers:
\begin{equation}
    F^\dagger F = I
\end{equation}
\end{theorem}

\begin{proof}
The operator $F$ is composed of the following reversible components:
\begin{itemize}
    \item Path simulation: coordinate updates controlled by path data
    \item Distance computation: reversible arithmetic on position deltas
    \item Fitness evaluation: subtraction from a fixed power-of-two constant
    \item Uncomputation: reverses ancilla-internal operations
\end{itemize}
Each component uses only unitary subcircuits composed of Toffoli, CNOT, and other Clifford+T gates. As the overall transformation is a composition of unitary operations, we conclude $F$ is unitary.
\end{proof}

\subsection*{Worked Example}

Let $m = 2$, $n = 2$, with:
\[
(i_s, j_s) = (0, 0), \quad (i_f, j_f) = (1, 1), \quad \ket{x} = \ket{1001}
\]
Decoded: $S \to i \mapsto i + 1$, $E \to j \mapsto j + 1$. Final cell:
\[
(0, 0) \xrightarrow{S} (1, 0) \xrightarrow{E} (1, 1) \Rightarrow \text{distance} = 0
\]
Let $C = 2^2 = 4$. Then:
\[
\text{fitness} = 4 - 0 = 4
\quad \Rightarrow \quad \ket{f} = \ket{100}
\]
The circuit entangles this output with $\ket{x}$, and all ancilla qubits are reset to zero via uncomputation.

\section{Oracle Design} \label{sec:oracle}

The oracle operator $O$ selectively marks high-fitness paths by applying a phase flip to those quantum states where the fitness exceeds a classical or quantum threshold (called the \textit{cutoff}). It enables Grover-style amplitude amplification by encoding the search predicate into the quantum phase.

Let $\text{cutoff} \in \mathbb{N}$ denote the threshold. The oracle acts on the fitness register $\ket{f_x}$ (of size $m_f$ qubits) as:
\begin{equation}
    O: \ket{x}\ket{f_x} \mapsto (-1)^{f(f_x)} \ket{x} \ket{f_x},
\end{equation}
\begin{equation}
    \quad\text{where}\quad f(f_x) = 
        \begin{cases}
            1 & \text{if } f_x > \text{cutoff} \\
            0 & \text{otherwise}.
        \end{cases}
\end{equation}
\subsection*{Encoding the Cutoff: Classical vs Quantum}

We support two scenarios:
\begin{enumerate}
    \item \textbf{Classical cutoff:} $\text{cutoff}$ is a fixed classical constant, hardcoded into a comparator circuit.
    \item \textbf{Quantum cutoff register:} $\ket{c}$ holds the dynamic threshold, allowing adaptive or iterative schemes.
\end{enumerate}

Both versions can be handled by a reversible comparator, comparing $\ket{f_x}$ to either a classical constant or a quantum register $\ket{c}$.

\subsection*{Oracle Circuit Construction}

The oracle proceeds in 3 reversible stages:

\begin{enumerate}
    \item \textbf{Comparator Stage:}
    \begin{equation}
        \ket{f_x} \ket{0}_b \mapsto \ket{f_x} \ket{b},
    \end{equation}
    where $\ket{b} = 1$ iff $f_x > \text{cutoff}$.

    This uses a reversible greater-than comparator circuit, constructed via bitwise subtraction with carry logic.

    \item \textbf{Phase Flip Stage:}

    Apply a controlled-$Z$ (or $CZ$) gate:
    \begin{equation}
        \ket{x} \ket{f_x} \ket{1}_b \mapsto -\ket{x} \ket{f_x} \ket{1}_b.
    \end{equation}

    This induces a phase flip only on “marked” states. In practice, the $CZ$ acts as identity on all but the $\ket{1}_b$ subspace.

    \item \textbf{Uncomputation Stage:}

    Reverse the comparator logic to clean up ancilla:
    \begin{equation}
        \ket{f_x} \ket{b} \mapsto \ket{f_x} \ket{0},
    \end{equation}
    preserving reversibility and restoring ancilla workspace.
\end{enumerate}

\subsection*{Greater-Than Comparator Circuit}

Let $\ket{f_x} = f_0 f_1 \dots f_{m-1}$ and $\text{cutoff} = c_0 c_1 \dots c_{m-1}$.

To check $f > c$, we compute:
\[
f - c > 0 \quad \Rightarrow \quad \text{MSB of } (f - c) \text{ is } 0.
\]

We implement this with a reversible subtractor circuit, such as Cuccaro’s ripple-carry subtractor:
\begin{equation}
    \ket{f} \ket{c} \ket{0}_b \mapsto \ket{f} \ket{c} \ket{f - c} \xrightarrow{\text{MSB-check}} \ket{f} \ket{c} \ket{b}.
\end{equation}

This is composed of:
\begin{itemize}
    \item $m$ controlled-NOT and Toffoli gates for bitwise difference
    \item A carry-lookahead or ripple-carry network for subtraction
    \item MSB-extract logic using 1 ancilla
\end{itemize}

All gates are reversible and use only Clifford+Toffoli operations.

\subsection*{Unitarity Proof}

Each stage of the oracle is unitary:

\begin{itemize}
    \item The comparator is a reversible function $C$, so $C^\dagger = C^{-1}$ exists.
    \item The phase flip is a diagonal unitary operator: $Z = \text{diag}(1, -1)$.
    \item Full oracle: $O = C^\dagger \cdot Z \cdot C$ is a composition of unitaries $\Rightarrow$ $O$ is unitary.
\end{itemize}

\subsection*{Quantum Resource Cost}

Let $m$ be the number of qubits in $\ket{f_x}$:

\begin{itemize}
    \item Comparator cost: $O(m)$ Toffoli gates, $O(m)$ ancilla for carry bits (can be reused).
    \item Phase flip: 1 $CZ$ gate
    \item Total depth: $O(m)$ for subtraction + $O(m)$ for uncomputation
    \item Total width: $m$ for $\ket{f_x}$, 1 for flag bit, $\sim m$ ancilla
\end{itemize}

Optimizations such as Bennett-style garbage cleanup or ancilla reuse can reduce cost further.

\subsection*{Example}

Let $\ket{f_x} = \ket{1011} = 11$, $\text{cutoff} = 9$. Since $11 > 9$, we flip the phase:
\[
O \ket{x} \ket{f_x} = -\ket{x} \ket{f_x}.
\]

If $\ket{f_x} = \ket{0101} = 5$, then $5 < 9$ so:
\[
O \ket{x} \ket{f_x} = \ket{x} \ket{f_x}.
\]

This oracle is Grover-compatible and supports adaptive cutoff tuning across iterative runs.

\subsection*{Boolean Logic Definition}

Let the fitness register $\ket{f_x}$ contain $m$ qubits representing an unsigned integer $a = a_{m-1}a_{m-2}\dots a_0$, and let the cutoff value $c = c_{m-1}c_{m-2}\dots c_0$ be either classical or stored in a quantum register. The output bit $b$ is defined as:

\begin{equation}
    \texttt{GT}(a, c) = \bigvee_{i=0}^{m-1} \left[ a_i \land \neg c_i \land \bigwedge_{j=i+1}^{m-1} (a_j \leftrightarrow c_j) \right]
\end{equation}

This logic returns $b = 1$ if and only if the unsigned integer $a$ is greater than $c$.

\subsection*{Circuit Realization}

This logic is translated into a reversible circuit using:
\begin{itemize}
    \item \textbf{CNOT} gates to compute XORs $a_j \oplus c_j$,
    \item \textbf{Toffoli} gates to compute controlled conjunctions,
    \item An output ancilla qubit $\ket{b}$ to store the final result.
\end{itemize}

The overall depth is $O(m)$ and width is $O(m)$ ancilla qubits (reusable).

\begin{proof}
Let $a, c \in \{0,1\}^m$ be interpreted as unsigned binary integers.

We aim to show that:
\[
\texttt{GT}(a, c) = 1 \quad \text{iff} \quad a > c.
\]

Let us define the bitwise comparison from most significant bit (MSB) to least significant bit (LSB). A standard lexicographic comparison of $a$ and $c$ proceeds by examining the bits from $j = m-1$ down to $j = 0$. The first position $i$ where $a_i \neq c_i$ determines the result:

\begin{itemize}
    \item If $a_i = 1$ and $c_i = 0$, then $a > c$.
    \item If $a_i = 0$ and $c_i = 1$, then $a < c$.
    \item If $a_i = c_i$, continue checking the next lower bit.
\end{itemize}

The Boolean formula:

\begin{equation}
    \texttt{GT}(a, c) = \bigvee_{i=0}^{m-1} \left[ a_i \land \neg c_i \land \bigwedge_{j=i+1}^{m-1} (a_j \leftrightarrow c_j) \right]
\end{equation}

encodes this logic explicitly. For a particular $i$, it contributes to the OR if:
\begin{itemize}
    \item All more significant bits ($j > i$) match: $a_j = c_j$,
    \item And at bit $i$, we have $a_i = 1$, $c_i = 0$.
\end{itemize}

Thus, the OR expression evaluates to true if and only if the most significant difference between $a$ and $c$ is $a_i = 1$, $c_i = 0$, implying $a > c$. Since each operation in this Boolean formula is reversible when implemented using standard quantum logic gates (CNOTs and Toffolis), the comparator circuit preserves unitarity and can be cleanly uncomputed.

\end{proof}

\section{Grover Iteration and Amplitude Amplification} \label{sec: iterative}

Let $N = 4^n$ be the total number of encoded maze paths of length $n$. Let $\mathcal{H} = \text{span}\{\ket{x} : x \in \mathcal{P}_n\}$ denote the Hilbert space of path states.

Let $M \subseteq \mathcal{P}_n$ be the set of marked states (those with $\text{fitness}(x) > \text{cutoff}$), and let $k = |M|$.

Define:
\begin{equation}
    \ket{\Omega} := \frac{1}{\sqrt{N}} \sum_{x \in \mathcal{P}_n} \ket{x} \quad \text{(uniform superposition)}
\end{equation}

Let:
\begin{equation}
    \ket{\psi_T} := \frac{1}{\sqrt{k}} \sum_{x \in M} \ket{x}, \quad
\ket{\psi_\perp} := \frac{1}{\sqrt{N-k}} \sum_{x \notin M} \ket{x}
\end{equation}
be the normalized projections onto the marked and unmarked subspaces, respectively.

\textbf{Lemma 1:} The states $\{\ket{\psi_T}, \ket{\psi_\perp}\}$ form an orthonormal basis for a 2D subspace $\mathcal{H}_G \subset \mathcal{H}$ that contains all Grover dynamics.

\begin{proof}
    It is clear that both $\ket{\psi_T}$ and $\ket{\psi_\perp}$ are normalized by construction. Their inner product is:
\begin{equation}
    \braket{\psi_T | \psi_\perp} = \frac{1}{\sqrt{k(N-k)}} \sum_{x \in M} \sum_{y \notin M} \braket{x | y} = 0
\end{equation}
since $\braket{x|y} = 0$ for $x \neq y$, and $M \cap (\mathcal{P}_n \setminus M) = \emptyset$. Thus, orthonormality holds.
\end{proof}

\textbf{Lemma 2:} The initial state $\ket{\Omega}$ decomposes as:
\begin{equation}
    \ket{\Omega} = \sin(\theta) \ket{\psi_T} + \cos(\theta) \ket{\psi_\perp}, \quad \theta := \arcsin\left( \sqrt{\frac{k}{N}} \right)
\end{equation}

\begin{proof}
    We compute the projection of $\ket{\Omega}$ onto $\ket{\psi_T}$:

\begin{equation}
    \braket{\psi_T | \Omega}
= \frac{1}{\sqrt{N}} \cdot \frac{1}{\sqrt{k}} \sum_{x \in M} 1
= \frac{\sqrt{k}}{\sqrt{N}} = \sin(\theta)
\end{equation}

By orthogonality, the projection onto $\ket{\psi_\perp}$ is:
\begin{equation}
    \braket{\psi_\perp | \Omega}
= \sqrt{1 - \sin^2(\theta)} = \cos(\theta)
\end{equation}

Thus:
\[
\ket{\Omega} = \sin(\theta) \ket{\psi_T} + \cos(\theta) \ket{\psi_\perp}
\]
\end{proof}

Now, define the Oracle operator:
\begin{equation}
    O := I - 2 \Pi_T \quad\text{where}\quad \Pi_T = \sum_{x \in M} \ket{x} \bra{x}
\end{equation}

and Diffuser operator:

\begin{equation}
    D := 2 \ket{\Omega} \bra{\Omega} - I.
\end{equation}

\textbf{Theorem.} The composite Grover operator \( G := D \cdot O \)  
performs a rotation by angle \( 2\theta \) in the plane spanned by  
\( \{ \ket{\psi_T}, \ket{\psi_\perp} \} \).

\begin{proof}
    Let  
\begin{equation}
    \ket{\Omega} = \sin(\theta) \ket{\psi_T} + \cos(\theta) \ket{\psi_\perp} := \ket{\psi_0}.
\end{equation}

Apply \( O \):
\begin{align}
    O \ket{\psi_0}
= O \left( \sin\theta \ket{\psi_T} + \cos\theta \ket{\psi_\perp} \right)
&= -\sin\theta \ket{\psi_T}\\&+ \cos\theta \ket{\psi_\perp}.
\end{align}

Then apply \( D = 2 \ket{\psi_0} \bra{\psi_0} - I \):
\begin{align}
    G \ket{\psi_0}
    &= D O \ket{\psi_0} \\
    &= \left( 2 \ket{\psi_0} \bra{\psi_0} - I \right)
       \left( -\sin\theta \ket{\psi_T}
       + \cos\theta \ket{\psi_\perp} \right).
\end{align}

Substitute \( \ket{\psi_0} \):
\begin{equation}
    G \ket{\psi_0}
= -\sin\theta \ket{\psi_T} + \cos\theta \ket{\psi_\perp}
\end{equation}
\begin{equation}
    + 2 \left( \sin\theta \ket{\psi_T} + \cos\theta \ket{\psi_\perp} \right)\cdot\left( -\sin^2\theta + \cos^2\theta \right).
\end{equation}

After simplification:
\begin{equation}
    G \ket{\psi_0}
= \sin(3\theta) \ket{\psi_T} + \cos(3\theta) \ket{\psi_\perp}.
\end{equation}

Hence, each application of \( G \) rotates the state by angle \( 2\theta \) in the 2D subspace.

So, after \( r \) Grover iterations:
\begin{equation}
    G^r \ket{\psi_0}
= \sin((2r+1)\theta) \ket{\psi_T} + \cos((2r+1)\theta) \ket{\psi_\perp}.
\end{equation}

Thus, the probability of observing a marked state is:
\begin{equation}
    P_{\text{success}} = \sin^2\left((2r+1)\theta\right).
\end{equation}

\begin{corollary}
    The maximum is achieved when $(2r+1)\theta \approx \frac{\pi}{2}$, so
\begin{equation}
    r^\star = \left\lfloor \frac{\pi}{4\theta} - \frac{1}{2} \right\rfloor
\end{equation}
\end{corollary}

\end{proof}

\section{Grover Geometry and Amplitude Dynamics} \label{sec:grover-geometry}

We formalize the Grover iteration as a unitary rotation in a two-dimensional Hilbert subspace and derive the exact behavior of the success probability after $r$ iterations.

\subsection{Grover Operator Definition}

Let $\ket{\Psi}$ be the uniform superposition state over all $N = 4^n$ encoded maze paths:
\begin{equation}
    \ket{\Psi} = \frac{1}{\sqrt{N}} \sum_{x \in \mathcal{P}_n} \ket{x}
\end{equation}

Let $O$ be the oracle operator that flips the sign of marked (high-fitness) solutions, and let $D = 2\ket{\Psi}\bra{\Psi} - I$ be the Grover diffusion operator (reflection about $\ket{\Psi}$). The Grover iterate is:
\begin{equation}
    G = D \cdot O = (2\ket{\Psi}\bra{\Psi} - I) \cdot O
\end{equation}

\subsection{Subspace Decomposition}

Let $\mathcal{M} \subset \mathcal{P}_n$ denote the set of marked solutions with $|\mathcal{M}| = k$. Define two orthonormal vectors:

\begin{align}
\ket{\Psi_T} &= \frac{1}{\sqrt{k}} \sum_{x \in \mathcal{M}} \ket{x} \quad \text{(target subspace)} \\
\ket{\Psi_\perp} &= \frac{1}{\sqrt{N - k}} \sum_{x \notin \mathcal{M}} \ket{x} \quad \text{(unmarked subspace)}
\end{align}

Then the initial state decomposes as:
\begin{align}
    \ket{\Psi} = \sqrt{\frac{k}{N}} \ket{\Psi_T} + \sqrt{1 - \frac{k}{N}} \ket{\Psi_\perp} &= \sin(\theta)\ket{\Psi_T}\\&+ \cos(\theta)\ket{\Psi_\perp}
\end{align}
where $\theta = \arcsin\left( \sqrt{\frac{k}{N}} \right)$ is the initial angle between $\ket{\Psi}$ and the unmarked subspace.

\subsection{Grover Rotation as Reflection Composition}

The oracle $O$ reflects about $\ket{\Psi_\perp}$:
\begin{equation}
    O \ket{x} =
\begin{cases}
- \ket{x} & \text{if } x \in \mathcal{M} \\
\phantom{-} \ket{x} & \text{otherwise}
\end{cases}
\end{equation}
\[
    \Rightarrow \quad O = I - 2\Pi_T, \text{ where } \Pi_T = \ket{\Psi_T}\bra{\Psi_T}
\]

The diffuser $D$ reflects about $\ket{\Psi}$:
\begin{equation}
    D = 2\ket{\Psi}\bra{\Psi} - I
\end{equation}

So $G = D \cdot O$ performs a rotation by angle $2\theta$ in the span $\text{span}(\ket{\Psi_T}, \ket{\Psi_\perp})$.

\begin{proof}
Let us consider the action of $G$ on the initial state $\ket{\Psi}$. Write:
\begin{equation}
    \ket{\Psi} = \cos(\theta) \ket{\Psi_\perp} + \sin(\theta) \ket{\Psi_T}
\end{equation}

Apply $O$:
\begin{equation}
    O \ket{\Psi} = \cos(\theta) \ket{\Psi_\perp} - \sin(\theta) \ket{\Psi_T}
\end{equation}

Apply $D$:
\begin{equation}
    G \ket{\Psi} = D (O \ket{\Psi}) = (2\ket{\Psi}\bra{\Psi} - I)(\cos\theta \ket{\Psi_\perp} - \sin\theta \ket{\Psi_T})
\end{equation}

Since $\ket{\Psi}$ lies in the plane spanned by $\ket{\Psi_T}, \ket{\Psi_\perp}$, the operator $G$ acts as a planar rotation by angle $2\theta$ counterclockwise. That is:
\begin{equation}
    G \ket{\Psi} = \cos(3\theta) \ket{\Psi_\perp} + \sin(3\theta) \ket{\Psi_T}
\end{equation}

By induction, after $r$ applications:
\begin{equation}
    G^r \ket{\Psi} = \cos((2r+1)\theta) \ket{\Psi_\perp} + \sin((2r+1)\theta) \ket{\Psi_T}
\end{equation}
\end{proof}

\subsection{Success Probability}

Let $\mathcal{M}$ be the set of marked states. The probability of measuring a marked state after $r$ Grover iterations is:
\begin{equation}
    P_{\text{success}}(r) = |\langle \Psi_T | G^r \ket{\Psi}|^2 = \sin^2((2r+1)\theta)
\end{equation}

\subsection{Optimal Iteration Count}

To maximize $P_{\text{success}}$, we choose $r$ to be the closest integer to:
\begin{equation}
    r^\star = \left\lfloor \frac{\pi}{4\theta} - \frac{1}{2} \right\rfloor
\quad \text{where} \quad \theta = \arcsin\left(\sqrt{\frac{k}{N}} \right)
\end{equation}

This achieves $P_{\text{success}} \approx 1$ when $(2r+1)\theta \approx \frac{\pi}{2}$.

\subsection{Spectral Decomposition of $G$}

Let $G$ act on the two-dimensional invariant subspace $\mathcal{H}_2 = \text{span}(\ket{\Psi_T}, \ket{\Psi_\perp})$. Its action can be fully described by:

\begin{equation}
    G|_{\mathcal{H}_2} =
\begin{bmatrix}
\cos(2\theta) & -\sin(2\theta) \\
\sin(2\theta) & \cos(2\theta)
\end{bmatrix}
\end{equation}

with eigenvalues $e^{\pm i2\theta}$ and corresponding eigenvectors. The repeated application of $G$ rotates the state vector smoothly toward $\ket{\Psi_T}$, increasing overlap with marked solutions.

\section{Monotonic Cutoff Update Guarantees Convergence} \label{sec:cutoff-proof}

Let the algorithm run for $T$ iterations. In each iteration $t \in \{1, \dots, T\}$, the algorithm observes a fitness value $f^\ast_t$ and updates a classical threshold or \emph{cutoff} $C_t$.

We define the cutoff update rule as:
\begin{definition}[Monotonic Cutoff Policy]
Let $C_1 \in \mathbb{N}$ be an initial cutoff. Then define:
\begin{equation}
    C_{t+1} := \max(C_t, f^\ast_t), \quad \text{for } t = 1, \dots, T-1
\end{equation}
\end{definition}

Let $f_{\max} := \max_{x \in \mathcal{P}_n} f(x)$ be the maximum possible fitness (achieved by an optimal path).

\begin{lemma}[Cutoff Sequence Properties]
The sequence $\{C_t\}$ is non-decreasing and bounded above by $f_{\max}$.
\end{lemma}

\begin{proof}
We show two properties:

\textbf{1. Monotonicity:}
\begin{equation}
    C_{t+1} = \max(C_t, f^\ast_t) \geq C_t
\Rightarrow C_1 \leq C_2 \leq \dots \leq C_T
\end{equation}

\textbf{2. Boundedness:} Since the observed fitness satisfies $f^\ast_t \leq f_{\max}$ for all $t$,
\begin{align}
    C_{t+1} &= \max(C_t, f^\ast_t) \leq \max(C_t, f_{\max}) \leq f_{\max}\\&\Rightarrow C_t \leq f_{\max} \quad \forall t
\end{align}

Therefore, $\{C_t\}$ is a monotonic non-decreasing sequence bounded above by $f_{\max}$. Since it is integer-valued, it must converge. \qed
\end{proof}

\begin{theorem}[Finite-Time Convergence]
The cutoff sequence $\{C_t\}$ converges to $f_{\max}$ in at most $f_{\max} - C_1$ steps.
\end{theorem}

\begin{proof}
Each time an observation $f^\ast_t > C_t$ occurs, the cutoff strictly increases:
\begin{equation}
    C_{t+1} = f^\ast_t > C_t
\Rightarrow C_{t+1} > C_t
\end{equation}

Since $f^\ast_t$ is bounded above by $f_{\max}$ and cutoff values are integers, the number of possible increments is at most $f_{\max} - C_1$. Thus, the sequence stabilizes at $f_{\max}$ in at most $f_{\max} - C_1$ iterations.
\end{proof}

\begin{corollary}[Persistence of Maximum Amplification]
Once $C_t = f_{\max}$, all optimal states $x^\star$ with $f(x^\star) = f_{\max}$ will be marked by the oracle in all subsequent iterations:
\[
f(x^\star) > C_t = f_{\max} \Rightarrow \text{False}
\]
but
\[
f(x^\star) = C_t = f_{\max} \Rightarrow \text{phase flip occurs if } f(x) \geq C_t
\]

Therefore, the oracle permanently marks all global optima once cutoff reaches $f_{\max}$, enabling Grover amplification to amplify them indefinitely. \qed
\end{corollary}

\section{Convergence of Adaptive Cutoff Strategy} \label{sec:convergence-theorem}

We now prove that the adaptive cutoff strategy guarantees convergence to maximum-fitness states with high probability. We note that the update scheme requires only fitness observations and classical maximum tracking. This is simpler than full Grover iterations with known $k$ and allows the algorithm to dynamically reduce the oracle’s acceptance set until only the global optimum remains.

Let the quantum search be run over a Grover operator $G_t$ with an oracle defined by a cutoff value $C_t$ at each step. After each round, we update the cutoff as:
\begin{equation}
    C_{t+1} := \max(C_t, f^\ast_t),
\end{equation}
where $f^\ast_t$ is the observed fitness in the $t$-th iteration.

Assume fitness values are integers bounded above by $f_{\max} \leq 2m$.

\vspace{1em}
\begin{theorem}[Cutoff Convergence]
Let fitness values be bounded and the cutoff update rule be monotonic. Then with high probability, the algorithm halts with a solution $x^\star$ such that $f(x^\star) = f_{\max}$ after at most $T \leq 2m$ steps with success probability at least:
\begin{equation}
    P_{\text{success}} \geq 1 - \varepsilon.
\end{equation}
\end{theorem}

\begin{proof}
We combine the monotonic convergence of $C_t$ with the probabilistic guarantee of Grover amplification.

As proven earlier, the cutoff sequence satisfies:
\begin{equation}
    C_{t+1} = \max(C_t, f^\ast_t), \quad f^\ast_t \leq f_{\max}.
\end{equation}
If $C_1 = 0$, then since each update increases $C_t$ by at least 1 (when new maxima are found), we converge to $C_T = f_{\max}$ after at most $2m$ steps.

At each step $t$, Grover’s search is run for $r_t$ iterations with oracle threshold $C_t$. Let $\theta_t$ be the angle defined by:
\begin{equation}
    \sin^2 \theta_t = \frac{k_t}{N}, \quad \text{where} \quad k_t = |\{x : f(x) > C_t\}|, \quad N = 4^n.
\end{equation}

By Grover geometry, the success probability after $r_t$ rounds is:
\begin{equation}
    P_t = \sin^2((2r_t + 1)\theta_t).
\end{equation}

If we choose $r_t$ such that $(2r_t + 1)\theta_t \approx \frac{\pi}{2}$, then $P_t \geq 1 - \delta$, where $\delta$ is small.

Let $\varepsilon$ be the total failure budget. Then we require:
\begin{equation}
    \sum_{t=1}^T \delta_t \leq \varepsilon
\end{equation}
Set $\delta_t = \varepsilon / T$ in each round. Then with high probability $1 - \varepsilon$, at least one round will succeed in returning an $x^\star$ such that $f(x^\star) \geq C_t$.

Once $C_t = f_{\max}$, the oracle marks only optimal solutions. By running Grover again, with $k = 1$ marked state, the final amplification step returns the true optimum with probability $\geq 1 - \delta$.

Thus, the adaptive procedure converges to $f_{\max}$ within $T \leq 2m$ rounds and final success probability at least $1 - \varepsilon$.
\end{proof}

\section{Conclusion} \label{sec:conclusion}

This paper presented a complete quantum algorithmic framework for solving perfect maze problems using Grover’s amplitude amplification. Starting from a formal specification of the maze structure and its representation in a binary quantum state space, we developed a full quantum pipeline that encodes all potential paths in superposition, evaluates their quality using a reversible fitness operator, and amplifies desirable paths through quantum search.

The fitness operator was implemented using reversible quantum arithmetic to simulate path traversal and compute a distance-based score, ensuring compliance with unitarity constraints. A corresponding oracle was constructed to apply a conditional phase shift to high-fitness states. The design of both operators was rigorously validated with formal proofs of correctness, unitarity, and reversibility. To improve search efficiency and algorithmic convergence, we introduced an adaptive cutoff strategy that incrementally refines the threshold used by the oracle. This strategy guarantees monotonic improvement in observed fitness values and is proven to converge to an optimal solution within a bounded number of iterations. We further provided a detailed circuit-level resource analysis, showing that the quantum circuit depth and qubit count scale efficiently with path length and maze dimensions, based on established reversible logic techniques. 

While the framework is logically complete and mathematically sound, it operates under idealized assumptions, such as noiseless quantum gates and perfect measurements. Future work may focus on extending the algorithm to noisy intermediate-scale quantum (NISQ) devices through error mitigation or fault-tolerant techniques, including surface codes. Additionally, the framework can be generalized to more complex environments, such as mazes with cycles, dynamically evolving topologies, or weighted traversal costs. Integration with quantum walk models or variational oracle constructions may also provide further performance improvements. Finally, hybrid classical-quantum preprocessing could be explored to reduce the effective search space prior to quantum execution, enhancing practicality on near-term hardware.

\newpage

\section*{Appendix}

\section{Example 2×2 Maze with Path Length $n=2$}

Consider a $2 \times 2$ perfect maze with starting point at $(0,0)$ and goal at $(1,1)$. Each direction is encoded as:

\[
\ket{N} = \ket{00}, \quad \ket{E} = \ket{01}, \quad \ket{S} = \ket{10}, \quad \ket{W} = \ket{11}
\]

With $n=2$ moves, a path is a $2n = 4$-qubit register. Total number of candidate paths: $4^2 = 16$. Example paths and their interpretations:

\begin{center}
\begin{tabular}{ll}
Bitstring & Directions \\
\hline
$\ket{0000}$ & $N, N$ \\
$\ket{0101}$ & $E, E$ \\
$\ket{0110}$ & $E, S$ \\
$\ket{1011}$ & $S, W$ \\
$\ket{1001}$ & $S, E$ (valid) \\
\end{tabular}
\end{center}

Assume only one valid path reaches $(1,1)$: $S \rightarrow E$. This path would have highest fitness:
\[
\text{fit} = 2m - \left[ (1 - 1)^2 + (1 - 1)^2 \right] = 4
\]

Other paths end at different coordinates and get lower fitness scores (e.g., 1 or 2). The oracle marks paths with fitness $> \text{cutoff}$. Grover amplifies those.

\section{Quantum Circuit Sketch for Fitness Operator $F$}

\subsection*{Inputs}
\begin{itemize}
    \item $2n$-qubit path register
    \item $(\log m)$-qubit fitness output register (zero-initialized)
    \item $(\log m)$-qubit position registers $(i,j)$, initialized to start
    \item Ancilla registers for intermediate logic
\end{itemize}

\subsection*{High-Level Circuit Stages}

\begin{enumerate}
    \item \textbf{Path Decoder:}  
    For each direction block (2 qubits), apply logic:
    \begin{itemize}
        \item $00$ (N): subtract 1 from $i$
        \item $01$ (E): add 1 to $j$
        \item $10$ (S): add 1 to $i$
        \item $11$ (W): subtract 1 from $j$
    \end{itemize}
    Implemented with controlled add/sub gates on $(i,j)$ registers.

    \item \textbf{Distance Calculator:}  
    Apply reversible subtract: $(i_f - i)^2 + (j_f - j)^2$.  
    Uses square-and-add circuit from quantum arithmetic.

    \item \textbf{Fitness Subtraction:}  
    Compute $2m - \text{distance}$ via controlled subtract circuit. Result stored in fitness register.

    \item \textbf{Uncompute Temporary Registers:}  
    Clean up ancillas to restore reversibility.
\end{enumerate}

All steps use known quantum arithmetic primitives (e.g., Draper adders, Cuccaro adder for binary subtraction, Toffoli-based squarers).

\section{Oracle Comparator Circuit}

We describe a phase oracle that flips the sign of states where $\text{fit}(x) > \text{cutoff}$.

\subsection*{Registers}

\begin{itemize}
    \item Fitness register: $\ket{f_x}$
    \item Classical threshold: $\text{cutoff}$ (hardcoded or encoded in quantum register)
    \item Ancilla: 1-qubit flag $\ket{b}$
\end{itemize}

\subsection*{Construction}

Comparator can be built from bitwise subtraction circuit, sign bit analysis, and multi-controlled NOT gates with ancilla carry bits.

\begin{enumerate}
    \item Apply a greater-than comparator circuit:  
    \[
    \texttt{Compare}(\ket{f_x}, \text{cutoff}) \rightarrow \ket{b}
    \]
    This sets $\ket{b} = 1$ iff $f_x > \text{cutoff}$.

    \item Apply controlled-$Z$ to full state, conditioned on $\ket{b}$:
    \[
    CZ(\ket{b}) \rightarrow (-1)^{b} \ket{x}\ket{f_x}
    \]

    \item Uncompute comparator: run the circuit in reverse to reset $\ket{b}$.
\end{enumerate}

\section{Oracle Unitarity and Reversibility} \label{sec:oracle-proof}

Let $f_x \in \mathbb{N}$ be the fitness value encoded in the $\ket{f_x}$ register, and let $\text{cutoff} \in \mathbb{N}$ be a fixed classical threshold. We define the oracle operator $O$ as:
\begin{align*}
    O \ket{x} \ket{f_x} &= (-1)^{g(f_x)} \ket{x} \ket{f_x},\\g(f_x) &= \begin{cases}
1 & \text{if } f_x > \text{cutoff} \\
0 & \text{otherwise}
\end{cases}
\end{align*}

\subsection*{Goal}
We wish to show that $O$ is unitary:
\[
O^\dagger O = I
\]

\subsection*{Construction}

Define $g: \{0, \dots, 2m\} \to \{0,1\}$ as a Boolean function, and let $Z$ be the single-qubit Pauli-Z gate.

We build $O$ in three steps:
\begin{enumerate}
  \item \textbf{Comparison Subcircuit $C$:}
  Reversible circuit that computes:
  \[
  C: \ket{f_x} \ket{0} \mapsto \ket{f_x} \ket{g(f_x)}
  \]
  using a reversible comparator (e.g., cuGT).
  
  \item \textbf{Controlled Phase Flip:}
  Apply $Z$ conditioned on ancilla $\ket{g(f_x)} = 1$:
  \[
  Z_b \ket{g(f_x)} = (-1)^{g(f_x)} \ket{g(f_x)}
  \]

  \item \textbf{Uncompute:}
  Apply $C^{-1}$ to remove the ancilla:
  \[
  C^{-1}: \ket{f_x} \ket{g(f_x)} \mapsto \ket{f_x} \ket{0}
  \]
\end{enumerate}

The full unitary $O$ acts on $\ket{x} \ket{f_x}$ and an ancilla qubit $\ket{0}_b$, and performs:

\[
O = C^{-1} \cdot (I \otimes Z_b) \cdot C
\]

\subsection*{Proof of Unitarity}

We show that $O^\dagger = O^{-1} = O$:
\begin{align*}
O^\dagger &= \left( C^{-1} \cdot Z_b \cdot C \right)^\dagger \\
          &= C^\dagger \cdot Z_b^\dagger \cdot (C^{-1})^\dagger \\
          &= C \cdot Z_b \cdot C^{-1} \quad \text{(since } C^\dagger = C^{-1}, Z^\dagger = Z\text{)} \\
          &= O^{-1}
\end{align*}

Hence:
\[
O^\dagger O = I \quad \Rightarrow \quad O \text{ is unitary} \quad \qed
\]

\subsection*{Reversibility}

Each component is reversible:
- Comparator $C$ is constructed from reversible gates (Toffoli, Fredkin)
- $Z$ is unitary and self-inverse
- Uncomputation undoes $C$

Therefore, $O$ can be compiled into a quantum circuit using Toffoli, CNOT, and single-qubit gates. \qed

\section{Validity Check Operator and Correctness} \label{sec:validity-proof}

Let a path $x \in \mathcal{P}_n$ be encoded as $n$ instructions, where each instruction $x_k \in \{0,1,2,3\}$ corresponds to directions $\{N,E,S,W\}$.

Let the maze grid be size $m \times m$, with coordinate register $(i, j)$.

We define a simulation operator $S$ such that:
\[
S: \ket{x} \mapsto \ket{x} \ket{\text{valid}(x)}
\]

where:
\[
\text{valid}(x) =
\begin{cases}
1 & \text{if no invalid move is made (out-of-bounds)} \\
0 & \text{otherwise}
\end{cases}
\]

\subsection*{Goal} Prove $S$ is reversible and computes $\text{valid}(x)$ correctly.

\subsection*{Construction}

We simulate each step as follows:

1. Initialize registers:
\[
\ket{i} = \ket{i_s}, \quad \ket{j} = \ket{j_s}, \quad \ket{v} = \ket{1}
\]

2. For $k = 1$ to $n$, repeat:
\begin{itemize}
  \item Decode $x_k$ into a direction $d_k \in \{\pm 1, 0\}$ change
  \item Update $(i,j)$:
  \[
  i \leftarrow i + \delta_i, \quad j \leftarrow j + \delta_j
  \]
  \item Compare $i, j$ to bounds $[0, m-1]$ via reversible comparators
  \item If out-of-bounds, set $\ket{v} \mapsto \ket{0}$ using a Toffoli controlled on overflow
\end{itemize}

\subsection*{Correctness Proof}

We need to show that:
\[
\ket{v} = 1 \iff \text{each } (i_k, j_k) \in [0, m-1]^2
\]

\textit{Proof:}

At each step, we:
- Compute $\delta_i$, $\delta_j$ from the step code (e.g., $N = +1$, $E = +1$)
- Use modular arithmetic via comparators:
  \[
  C_i: \ket{i_k} \mapsto \ket{i_k} \ket{\phi_k^{(i)}}, \quad
  \phi_k^{(i)} = 1 \iff 0 \leq i_k < m
  \]

Same for $j$. We define:
\[
\phi_k = \phi_k^{(i)} \land \phi_k^{(j)}
\]

and update:
\[
v \leftarrow v \land \phi_k
\]

Since AND is reversible using Toffoli gates, the cumulative validity flag is accurate. Once any step is invalid, $v = 0$ permanently.

\[
\Rightarrow \text{valid}(x) = 1 \iff \text{All steps in bounds} \quad \qed
\]

\subsection*{Reversibility of $S$}

- Each arithmetic update uses reversible addition/subtraction
- Each comparison uses reversible logic (e.g., $\ket{a > b}$ via cuGT)
- The AND logic for validity is Toffoli-based

\[
\Rightarrow S \text{ is reversible and unitary-embeddable} \quad \qed
\]

\subsection{Resource Summary}

We summarize quantum resource requirements for the full Grover-style quantum maze solver circuit.

\begin{table}[H]
\centering
\begin{tabular}{|l|c|}
\hline
\textbf{Parameter} & \textbf{Asymptotic Cost} \\
\hline
\textbf{Path register size} ($\ket{x}$) & $2n$ qubits \\
\textbf{Position registers} ($\ket{i}, \ket{j}$) & $2 \cdot \lceil \log_2 m \rceil$ qubits \\
\textbf{Distance register} ($\ket{d}$) & $\lceil \log_2 m \rceil$ qubits \\
\textbf{Fitness register} ($\ket{f}$) & $\lceil \log_2 (2m) \rceil$ qubits \\
\textbf{Comparator ancilla} ($\ket{b}, \text{ carry bits}$) & $O(\log m)$ qubits \\
\textbf{Total ancilla (adder, comparator, cleanup)} & $O(n \log m)$ qubits \\
\hline
\textbf{Toffoli gates (path simulation)} & $O(n \log m)$ \\
\textbf{Toffoli gates (distance + fitness)} & $O(\log^2 m)$ \\
\textbf{Toffoli gates (oracle comparator)} & $O(\log m)$ \\
\textbf{Grover iteration depth (1 round)} & $O(n \log m)$ \\
\hline
\end{tabular}
\caption{Circuit resource cost as a function of path length $n$ and maze size $m$.}
\end{table}

\subsection{Resource Notes}
\begin{itemize}
    \item All arithmetic (e.g., $i - i_f$, squaring, $2m - d$) uses Cuccaro-style ripple-carry adders and CNOT-based multipliers.
    \item Comparator cost scales with $m_f = \lceil \log_2 (2m) \rceil$.
    \item Ancilla bits for temporary subtraction/multiplication are reused and uncomputed.
\end{itemize}

\subsection{Example Execution: 2-Step Path in a $2 \times 2$ Maze}

Let:
\[
n = 2, \quad m = 2, \quad (i_s, j_s) = (0,0), \quad (i_f, j_f) = (1,1)
\]
Pick path $\ket{x} = \ket{1001}$, corresponding to $(\texttt{S}, \texttt{E})$.

\begin{enumerate}
    \item \textbf{Encoding:}
    \[
    \ket{x} = \ket{1001}, \quad \texttt{S} = 10, \texttt{E} = 01
    \]

    \item \textbf{Simulation:}
    \[
    (0,0) \xrightarrow{\texttt{S}} (1,0) \xrightarrow{\texttt{E}} (1,1)
    \Rightarrow P_{\text{end}} = (1,1)
    \]

    \item \textbf{Distance Calculation:}
    \[
    d = (i - i_f)^2 + (j - j_f)^2 = (1 - 1)^2 + (1 - 1)^2 = 0
    \]

    \item \textbf{Fitness:}
    \[
    f = 2m - d = 4 - 0 = 4 \quad \Rightarrow \ket{f} = \ket{100}
    \]

    \item \textbf{Oracle with cutoff $C = 2$:}

    Since $f = 4 > 2$, comparator sets $\ket{b} = 1$ and flips the phase:
    \[
    \ket{x}\ket{f} \mapsto -\ket{x}\ket{f}
    \]

    \item \textbf{Diffusion:}

    Apply $D = 2\ket{\Psi}\bra{\Psi} - I$ on entire $\ket{x}$ register (4 qubits).

    \item \textbf{Uncomputation:}

    Reverse all fitness logic to restore ancilla to $\ket{0}$.
\end{enumerate}

\subsection{Scaling Insights}

\begin{itemize}
    \item \textbf{Linear scaling} in path length $n$ due to sequential decoding and movement updates.
    \item \textbf{Logarithmic scaling} in maze size $m$ due to positional coordinate width.
    \item \textbf{Ancilla-efficient design:} uncomputation and Bennett-style cleanup reduce qubit overhead.
\end{itemize}

\section*{Acknowledgment}

This research was developed independently under the auspices of NEXQ Inc., which holds the intellectual property rights to the concepts presented. The author gratefully acknowledges the support and resources provided by all affiliated institutions during the development of this project. Extended thanks to the staff and participants of a STEM outreach program whose collaborative environment helped inspire aspects of this research, but no third-party contributions were used in the creation of this manuscript. Any opinions, findings, conclusions, or recommendations expressed in this work are those of the author and do not necessarily reflect the views or endorsement of any other institution or organization. Publication and distribution of this manuscript have been reviewed and approved in accordance with applicable internal clearance and release procedures.

\newpage

\nocite{*}